\documentclass{article}

\usepackage{kerger}
\usepackage{wrapfig}

\let\dist\thicksim

\usepackage{hyperref}
\usepackage[numbers]{natbib}
\usepackage{doi}
\usepackage{amsthm}
\setcitestyle{numbers}

\usepackage{etoolbox}
\newcounter{bibcount}
\makeatletter
\patchcmd{\@lbibitem}{\item[}{\item[\hfil\stepcounter{bibcount}{\thebibcount.}}{}{}
\renewcommand\NAT@bibsetup%
   [1]{\setlength{\leftmargin}{\bibhang}\setlength{\itemindent}{-\parindent}%
       \setlength{\itemsep}{\bibsep}\setlength{\parsep}{\z@}}
\makeatother
\usepackage[margin=1cm]{caption}

\usepackage{authblk} 

\begin{document}

\title{Quantum Image Denoising: A Framework via Boltzmann Machines, QUBO, and Quantum Annealing}

\author[1,2,3]{Phillip Kerger}
\author[4,5]{Ryoji Miyazaki }
        \affil[1]{Department of Applied Mathematics and Statistics, Johns Hopkins University.}
        \affil[2]{Quantum Artificial Intelligence Laboratory, NASA Ames Research Center, Moffett Field CA 94035, USA.}
        \affil[3]{Research Institute of Advanced Computer Science, USRA, Moffett Field CA 94035, USA.}
        \affil[4]{
        Secure System Platform Research Laboratories, NEC Corporation, Kawasaki, Kanagawa 211-8666, Japan.
        }
\affil[5]{NEC-AIST Quantum Technology Cooperative Research Laboratory, National Institute of Advanced Industrial Science and Technology
(AIST), Tsukuba, Ibaraki 305-8568, Japan.}

\date{\today}

\maketitle

\begin{abstract}
    We investigate a framework for binary image denoising via restricted Boltzmann machines (RBMs) that introduces a denoising objective in quadratic unconstrained binary optimization (QUBO) form and is well-suited for quantum annealing. The denoising objective is attained by balancing the distribution learned by a trained RBM with a penalty term for derivations from the noisy image. We derive the statistically optimal choice of the penalty parameter assuming the target distribution has been well-approximated, and further suggest an empirically supported modification to make the method robust to that idealistic assumption. We also show under additional assumptions that the denoised images attained by our method are, in expectation, strictly closer to the noise-free images than the noisy images are. While we frame the model as an image denoising model, it can be applied to any binary data. As the QUBO formulation is well-suited for implementation on quantum annealers, we test the model on a D-Wave Advantage machine, and also test on data too large for current quantum annealers by approximating QUBO solutions through classical heuristics.
\end{abstract}

\section{Introduction}\label{sec: introduction}
Quantum annealing (QA)~\cite{T.Kadowaki1998, A.Das2008, T.Albash2018Jan} is a promising technology for obtaining good solutions to difficult optimization problems, by making use of quantum interactions to aim to solve Ising or quadratic unconstrained binary optimization (QUBO) instances. Since Ising and QUBO instances are NP-hard, and many other combinatorial optimization problems can be reformulated as Ising or QUBO instances (see e.g. \cite{QuboFormulationTutorial_glover2018}), QA has the potential to become an extremely useful tool for optimization. As the capacities of commercially available quantum annealers continue to improve rapidly, it is of great interest to build models that are well-suited for this emerging technology. Furthermore, QA has promising machine learning applications surrounding Boltzmann Machines (BMs), as both QA and BMs are closely connected to the Boltzmann distribution. Boltzmann Machines are a type of generative artificial neural network that aim to learn the distribution of some training data set by fitting a Boltzmann distribution to the data, as described thoroughly in \cite[\S 20]{deep_learning_bengio}. On the other hand, QA aims to produce approximate minimum energy (maximum likelihood) solutions to a Boltzmann distribution via finding the ground state of the associated Hamiltonian that determines the distribution. Hence, maximum likelihood type problems on BMs are a natural candidate for applying QA in a machine learning framework. 
We contribute to the goal of furthering useful applications of QA in machine learning in this paper by building an image denoising model particularly well-suited for implementation via QA. 

The task of image denoising is a fundamental problem in image processing and machine learning. In any means of collecting images, there is always a chance of some pixels being afflicted by noise that we wish to remove; see e.g. \cite{Noise_models_in_images_review} for a good overview. Accordingly, many classical and data-driven approaches to the image denoising problem have been studied in the literature \cite{image_denoising_survey_buades2005, TangHinton2012RobustBM_denoising, Greig1989_graph_cut_denoising, rudin1992nonlinear, BM_denoising}. This paper studies a quantum binary image denoising model using Restricted Boltzmann Machines (RBMs henceforth) \cite[\S 20.2]{deep_learning_bengio} that can take advantage of QA by formulating the denoising problem as a QUBO instance. Specifically, given a trained RBM, we introduce a penalty-based denoising scheme that admits a simple QUBO form, for which we derive the statistically optimal penalty parameter as well as a practically-motivated robustness modification. The denoising step only needs to solve a QUBO admitting a bipartite graph representation, and so is well-suited for QA. 
As QA has also shown promise for training BMs \cite{adachi2015_QA_for_DeepNets, Dixit_2021}, our full model lends itself well for denoising images using quantum annealers, and could thus play a role in the their future applications since QA can then be leveraged for \textit{both }the training and denoising steps. The model also shows promise in absence of QA, and our insights presented are not limited to the QA framework, as the QUBO formulation of the denoising problem and its statistical properties we prove may be of independent interest.

The paper is organized as follows. Section \ref{background} gives a summary of background on quantum annealing and Boltzmann Machines. Section \ref{sec: methods} describes our main contribution of the image denoising model for QAs, and Section \ref{results} shows some practical results obtained. 

\begin{remark}
    We frame our work as a binary image denoising method, although the framework does not depend on the data being images, and can be applied to the denoising of any binary data. This is because the framework does not use any spatial relationships between the pixels, and instead treats the image as a flattened vector whose distribution is to be learned. Hence, the denoising scheme can be applied as-is to any other binary data setting. 
\end{remark}

\subsection{Contributions and Organization}
We provide QUBO-based denoising method for binary images (applicable to general binary data) using restricted Boltzmann machines in Section \ref{sec: methods}. This is done by formulating the denoising objective in equation \ref{eq: penalty model} by combining the energy function of the distribution learned by the RBM with a (parameterized) penalty term for deviations from a given noisy image. This objective turns out to have an equivalent QUBO formulation, which is shown in claim \ref{claim:penalty model is qubo}. In Theorem \ref{thm: rho optimality}, we derive the optimal choice for the penalty parameter under the assumption that the true images follow the distribution learned by the RBM, which also recovers the maximum a posteriori estimate per Corollary \ref{cor: map recovery}, though our model is more flexible, and this flexibility allows for useful practical modifications. Theorem \ref{thm: provable denoising} shows that the denoising method yields a result that is {\em strictly} closer (in expectation) to the true image than the noisy image is, under some additional assumptions. Given that these idealistic assumptions won't be met in reality, we propose a robustness modification in Section \ref{sec: robust rho} that \textit{improves} performance empirically. In Section \ref{results}, as the method lends itself well to quantum annealing, we then implement the method on a D-Wave Advantage 5000-qubit quantum annealer, demonstrating strong empirical performance. Since only small datasets can be tested on the D-Wave machine due to the relatively low number of qubits, we also test the method on a larger dataset, for which we use simulated annealing on a conventional computer in place of quantum annealing to find good solutions the QUBO denoising objective. 
Though we highlight the method being well-suited for quantum annealers, we emphasize that it may be of independent interest to the machine learning and image processing communities at large. 

\subsection{Related Work}
Closely related work of \cite{koshka_reconstruction} uses a similar model as ours for the image reconstruction task, also solving QUBO formulations via quantum annelaing. In the reconstruction task, some subset of pixels is unknown (or obscured or missing), and needs to be restored, whereas our work considers denoising, where which pixels are noise-afflicted is unknown. \cite{Greig1989_graph_cut_denoising} derives a maximum a posteriori (MAP) estimator for the noise free image as a denoising method in a particular model of binary images that is less general than ours, though we would recover their estimator under a particular choice of our penalty parameter if we were to apply our framework to their model (since we recover MAP in a more general setting). Further, RBMS and quantum annealing have been studied for the classification problem, for instance in \cite{Krzysztof2021_QA_RBM_for_MNIST_classification} and \cite{adachi2015_QA_for_DeepNets}.
Other research in the machine learning communities has also studied handling \textit{label noise}, such as related work in \cite{Vahdat2017TowardRA_learning_from_noisy}, which studies the problem of training models in the presence of noisy labels, whereas our approach is entirely unsupervised (the data need not have any labels to begin with).

\begin{section}{Background}\label{background} 
Quantum Annealers make use of quantum interactions with the primary goal of finding the ground state of Hamiltonian by initializing and then evolving a system of coupled qubits over time~\cite{M.Johnson2011}. In particular, we may view QA as implementing the Ising spin-glass model \cite{StatisticalPhysics} evolving over time.
As the QUBO model is equivalent to the Ising model \cite{QuboFormulationTutorial_glover2018}, and QUBO instances can be efficiently transformed to Ising instances, QA is well suited to provide good solutions to QUBO problems. A QUBO cost function, or energy function, takes the form
\begin{align}
f_Q(x) := \sum_{i, j} Q_{ij}x_i x_j
\label{eq:cost}
\end{align}
where $x_i \in \{0,1\}$, and $Q$ is a symmetric, real-valued matrix. We will occasionally refer to $Q_{ij}$ as the $weight$ between $x_i$ and $x_j$. QUBO is well-known to be NP hard~\cite{F.Barahona1982}, and many combinatorial problems can be reformulated as QUBO instances. See  \cite{QuboFormulationTutorial_glover2018, A.Lucas2014} for thorough presentation of QUBO formulations of various problems. A Boltzmann Distribution using the above QUBO as its energy function takes the form 
\begin{align}
P_Q^{model}(x) = \frac{1}{z} \exp(-f(x,Q)), 
\label{eq:bm prob}
\end{align}
where $z$ is a normalizing constant. 
Note that a parameter called inverse temperature has been fixed to unity and is not explicitly shown in the above expression.
In this paper, we will focus on making use of Boltzmann Machines, a type of generative neural network that fits a Boltzmann Distribution to the training data via making use of latent variables. Specifically, we consider Restricted Boltzmann Machines (RBMs), which have seen significant success and frequent use in deep probabilistic models \cite{deep_learning_bengio}. RBMs consist of an input layer of $visible$ nodes, and a layer of latent, or $hidden$ nodes, which each have zero intra-group weights. Let $\vv \in \{0,1\}^v$ and $\h \in \{0,1\}^h$ denote the visible and hidden nodes, respectively. It will be convenient for us to write $x = (\vv,\h) \in \{0,1\}^{v+h}$ as their concatenation. The probability distribution represented by a RBM is then
\begin{align}
P_Q^{model}((\vv, \h)) = \frac{1}{z}\exp(-f((\vv, \h), Q))
\label{eq:likelihood_model}
\end{align}
with the restriction that $Q_{ij} = Q_{ji} = 0$ if 
$i, j \in \{ 1, \dots, v\}$ or $i, j \in \{ v+1, \dots, v+h \}$.
Hence, we have the simplified energy function 
\begin{align} \nonumber
    f((\vv,\h), Q) &= \sum_{i=1}^{v+h}  \sum_{j=1}^{v+h}
    2Q_{ij}(\vv, \h)_i(\vv, \h)_j = \sum_{i=1}^{v} \sum_{j=v+1}^{v+h} Q_{ij} \vv_i \h_j + \sum_{i=1}^{v} Q_{ii} \vv_i^2 + \sum_{i=v+1}^{v+h} Q_{ii}\h_i^2 \\
    &= \h^T W \vv + b_v^T\vv + b_h^T \h
    =: f_{W, b_v,b_h}(\vv, \h) \label{eq: convenient energy f}
\end{align}
where $W$ is the $v \cross h$ matrix consisting of the $Q_{ij}$ weights between the visible and hidden nodes, and $b_v$ and $b_h$ are vectors of the diagonal entries $Q_{ii}$, $i\in \{ 1, \dots, v \}$ corresponding to visible nodes,  and $Q_{ii}, i \in \{n+1, ..., v+h\}$ corresponding to hidden nodes, respectively. We will write the Boltzmann distribution with this energy function as $P_{W, b_v, b_h}$, noting that this is also $P_Q^{model}$ for the appropriate $Q$. \\

 It is well known that RBMs can universally approximate discrete distributions \cite{deep_learning_bengio}, making them a powerful model. They are also more easily trained than general Boltzmann Machines, usually through the contrastive divergence algorithm as described in \cite{contrastive_divergence_Hinton2002}, or variants thereof.

 \subsection{Training Boltzmann Machines}
We first devote some discussion to the training of RBMs. Subsection \ref{subsec: qubo denoise} then describes how to denoise images via QUBO given a well-trained RBM. \\

Continuing with the notation as in equation \ref{eq: convenient energy f}, the probability distribution represented by a RBM is 
$$
P_{\theta}(\vv, \h) 
= \frac{1}{z_\theta} \exp(-f_{\theta}).
$$
For simplicity, denote $\theta = (W, b_v, b_h)$ as the model parameters henceforth. The normalizing constant $z_\theta$ above is $$
z_\theta = \sum_{\vv \in \{0,1\}^v} \sum_{\h\in\{0,1\}^h} \exp(-f_\theta(\vv, \h))
$$
which is becomes intractable quickly even for relatively small values of $v$ and $h$. 
The common training approach aims to maximize the log-likelihood of the data. At a high-level, this will be done by approximating gradients and following a stochastic gradient scheme. 
However, since our data consists only of the visible nodes, we need to work with the marginal distribution of the visible nodes. This is given by 
$$
P_{\theta}(\vv) = \sum_\h P_{\theta}(\vv, \h) = \sum_\h \frac{ \exp [- f_\theta(\vv, \h) ] }{z_\theta}
$$
Denote our set training data samples by $V := \{\vv^1, ..., \vv^N\}$. We will use superscripts to indicate training data samples, and reserve subscripts to denote entries of vectors. Then the log-likelihood is given by 
\begin{align}
    \nonumber 
    l_{\theta}(V) &= \sum_{k=1}^N \log P_{\theta}(\vv^k) = \sum_{k=1}^N \log \sum_{\h} P_{\theta}(\vv^k,\h) \\ 
    \nonumber
    &= \left ( \sum_k \log \sum_\h \exp (-f_\theta(\vv^k,\h)) \right ) - N \cdot  \log z_\theta \\
    &= \left ( \sum_k \log \sum_\h \exp (-f_\theta(\vv^k,\h)) \right ) - N \cdot \log \sum_{\vv} \sum_{\h} \exp(-f_\theta(\vv, \h)) \label{eq: log likelihood}
\end{align}

Now we can calculate the gradient with respect to $\theta$ as 
\begin{align*}
    \nabla l_\theta(V) 
    &= \sum_{k=1}^N \dfrac{\sum_\h \exp (-f_\theta(\vv^k,\h))\nabla(-f_\theta(\vv^k, \h))}{\sum_\h \exp (-f_\theta(\vv^k,\h))} - N\cdot \dfrac{\sum_{\vv, \h} \exp (-f_\theta(\vv, \h))\nabla(-f_\theta(\vv, \h))}{\sum_{\vv, \h} \exp (-f_\theta(\vv, \h))} \\
    &= \sum_{k=1}^N \mathbb{E}_{P_\theta(\h|\vv^k)} \left[ -\nabla f_\theta(\vv^k, \h) \right] -N\cdot \mathbb{E}_{P_{\theta}(\vv, \h)}\left[  -\nabla f_\theta(\vv, \h) \right] \\
    &= \frac{1}{N} \sum_{k=1}^N \mathbb{E}_{P_\theta(\h|\vv^k)} \left[(\vv^k)^T \h + \vv^k + \h  \right] 
    - \mathbb{E}_{P_{\theta}(\vv, \h)}\left[ \vv^t\h + \vv +\h \right] 
\end{align*}
The first term can be computed exactly and efficiently from the data, since the conditional $P_\theta(\h|\vv)$ admits the simple form $P(\h_j = 1|\vv) = logistic(b_h + (\vv^TW)_j)$; we refer the interested reader to \cite{Dixit_2021} or \cite{deep_learning_bengio} and will focus on the second term. 
Due to its intractability to compute (one would have to sum over all possibilities of $\vv$ and $\h$), the most promising approach is to approximate it by sampling from $P_\theta(\vv, \h)$. Classically, this is done via Gibbs sampling as described in \cite{contrastive_divergence_Hinton2002}. 
However, recent research has also investigated using quantum annealers to sample from the relevant Boltzmann distribution, as suggested in \cite{Dixit_2021}, which would make QAs useful in the training process since obtaining good Gibbs samples can be expensive. We note that together with our framework, QAs show promise to become useful for both the RBM training and the denoising process in the implementation of our method. 
\\
\end{section}

\section{Image Denoising as Quadratic Unconstrained Binary Optimization}\label{sec: methods}
This section is devoted to showing how one can naturally frame the image denoising problem as a QUBO instance over a learned Boltzmann Distribution fit to the data.

\subsection{Denoising via QUBO} \label{subsec: qubo denoise}
Let us assume we are given a trained Restricted Boltzmann Machine 
described in Sec.~\ref{background}.
The model prescribes to each vector $x\in \{0,1\}^{v+h}$ 
the cost $f_Q(x)$ and corresponding likelihood $P_{Q}^{model}(x)$ 
defined in Eqs.~(\ref{eq:cost}) and (\ref{eq:likelihood_model}), respectively.
We will here make the assumption that $P_{Q}^{model}$ describes the distribution of our data. Hence, high likelihood vectors in $P_{Q}^{model}$ correspond to low cost vectors of $f_Q$. In particular, note that finding the maximum likelihood argument in (\ref{eq:bm prob}) corresponds to finding a solution to the QUBO instance in (\ref{eq:cost}). \\
Now, supposing this model, our goal is to reconstruct an image that has been affected by noise. The visible portion of our vector will be considered to be a flattened image with $v$ pixels, black or white corresponding to 0 or 1, respectively, in the binary entries of the vector.

\subsubsection{Noise Model}
We now describe the noise assumptions we will conduct our analysis under. 
\begin{definition}
For $x\in\{0,1\}^v$, we define $x$ {\em afflicted by salt-and-pepper noise of level $\sigma$} as the random variable $\tilde{X}_{x, \sigma}:= (x + \epsilon) mod 2 $, where $\epsilon_i = B_{i}(p) \dist Bern(\sigma)$, independently.
\end{definition}
In other words, a binary image afflicted by salt-and-pepper noise has each pixel independently flipped with probability $\sigma$. In particular, we are interested in $\tilde{X}_{X, \sigma}$, where $X\dist  P_{Q}^{model}$, which is the compound random variable obtained by sampling $X$ from the learned distribution of the data and then afflicting it with salt-and-pepper noise. For notational simplicity, will simply write $\tilde{X}$ when the intended subscripts are clear from context. Note that salt-and-pepper noise is a natural noise model for binary data, since the only means in which pixels (or data entries, for general binary data) can be changed is by flipping the $0-1$ value.

Suppose we are given a realization $\tilde{x}\in \{0,1\}^v$ of $\tilde{X}_{X, \sigma}$.
The reconstruction process aims to retrieve this original $X$ using $\tilde{x}$ and the trained model through $Q$. 
The approach we will take begins from the intuition that $X$ is likely to be a high-likelihood image that is close to $\tilde{x}$.
To enforce this "closeness" to $\tilde{x}$ while searching for higher likelihood images in our model to remove noise, we add to the cost in (\ref{eq:cost}) a penalty for deviations from $\tilde{x}$ to formulate the following natural denoising cost function: 
\begin{align}
    f_{Q, \tilde{x}, \rho}(x) = f_Q(x) + \rho \sum_{i, j}(x_i - \tilde{x}_i)^2
    \label{eq: penalty model}
\end{align}
for some $\rho > 0$ that determines the penalty level. 
The intuition is that the minimizer of this function for a well-chosen $\rho$ will change a restricted number of pixels to find an image that is similar to the noisy image, but has a lower cost, i.e. higher likelihood, under the model, in hopes of removing the noise. 

We show next that this minimizing (\ref{eq: penalty model}) corresponds to solving a QUBO instance. 
\begin{claim}\label{claim:penalty model is qubo} Defining $\tilde{Q}^{\rho, \tilde{x}} \in \mathbb{R}^{(v+h)\times (v+h)}$ by setting $\tilde{Q}_{ij}^{\rho, \tilde{x}} = Q_{ij}$ if ${i \neq j}$ and $\tilde{Q}_{ij}^{\rho, \tilde{x}} = Q_{ii} + \rho(1-2\tilde{x})$ if $i = j$, we have
\begin{align}
    argmin_x f_{Q, \tilde{x}, \rho}(x) = argmin_x f_{\tilde{Q}^{\rho, \tilde{x}}}(x).
    \label{eq: penalty reconstruct}
\end{align}
\end{claim}

\begin{proof}
\begin{align*}
   f_{Q, \tilde{x}, \rho} (x) &= f_Q(x) + \rho \sum_{i}(x_i - \tilde{x}_i)^2 
    = \sum_{i, j} Q_{ij}x_ix_j + \rho \sum_i x_i^2 - 2x_i\tilde{x}_i + \tilde{x}_i^2 \\
    &= \sum_{i \neq j} Q_{ij}x_ix_j + 
    \sum_i Q_{ii}x_i^2 + \rho (x_i^2 - 2x_i^2\tilde{x}_i + \tilde{x}^2) \\
    &= \sum_{i \neq j} Q_{ij}x_ix_j 
    + \sum_i (Q_{ii} + \rho(1 - 2\tilde{x}_i)) x_i^2 + \rho\tilde{x}_i^2 
    = f_{\tilde{Q}^{\rho, \tilde{x}}}(x) + \sum_i \rho\tilde{x_i}
\end{align*}
 Noting  that $x_i = x_i^2$ for the above derivation since they are in $\{0,1\}$ here. Since the $\tilde{x}_i$ terms do not depend on $x$, the claim follows.
\end{proof}
Hence, solving the QUBO in on the right hand side of equation \ref{eq: penalty reconstruct} gives us the solution to \ref{eq: penalty model}. 
Claim \ref{claim:penalty model is qubo} thus tells us that we simply need to modify the diagonal of the original matrix $Q$ of our model by adding $diag(1-2\tilde{x}_1, ..., 1-2 \tilde{x}_n)$ and then solve the resulting QUBO to get the denoised image. 
We can then make use of quantum annealing to solve the resulting QUBO of \ref{eq: penalty reconstruct}, or use classical methods and heuristics like simulated annealing instead. 
We formally spell out the denoising procedure in algorithm \ref{alg: main denoising}. 
\begin{mdframed}\namedlabel{alg: main denoising}{\texttt{QUBO\_Denoise}}\\
Input: A matrix $Q$, a noisy image $\tilde{x}$ sampled from the distribution of $\tilde{X}_{X, \sigma}$ with $X\dist P_{Q}^{model}$, and a penalty parameter $\rho >0$.\\
Output: A denoised image $X^*_{\rho, \tilde{x}, Q}$.
\begin{enumerate}
    \item Set $\tilde{Q}_{ij}^{\rho, \tilde{x}} = Q_{ij}$ if ${i \neq j}$ and $\tilde{Q}_{ij}^{\rho, \tilde{x}} = Q_{ii} + \rho(1-2\tilde{x})$ if $i = j$.
    \item Set $X^*_{\rho, \tilde{x}, Q} := argmin_x f_{\tilde{Q}^{\rho, \tilde{x}}}(x)$.
\end{enumerate}
\end{mdframed}

For the remainder of the paper, $X^*_{\rho, \tilde{x}, Q}$ will denote the denoised image obtained by applying \ref{alg: main denoising} with noisy image $\tilde{x}$, penalty parameter $\rho$, and the distribution-defining matrix $Q$.

\begin{remark}
 Considering the entire process of sampling a noisy image and then denoising it, the measurability of $X^*_{\rho, \tilde{X}_{X, \sigma}, Q}$ is inherited from the measurability of $\tilde{X}_{X, \sigma}$, which in turn inherits its measurability as compound random variable of the measurable noise and original image $X \dist P_{Q}^{model}$.
\end{remark}

\subsection{Optimal Choice of penalty parameter $\rho$ }
The choice of the parameter $\rho$ for the proposed image denoising model is clearly crucial to its success, since different choices will result in different solutions. 
If $\rho$ is chosen to be too small, there is very little cost to flipping a pixel, and then many pixels may be flipped and the solution may not resemble the noisy image at all anymore. If $\rho$ is too large, we may be too heavily penalizing flipping pixels, and thus may not be able to get rid of noise effectively. Hence, we now turn towards finding the optimal choice for $\rho$. We will evaluate the choice of $\rho$ via {\em expected overlap}: 
\begin{definition}
The \textit{expected overlap} between two distributions $P$ and a $P'$, is defined by 
$$
d(P, P') := \mathbb{E}_P \mathbb{E}_{P'}\left[n- \norm{X - X'}_1\right], $$
where $X\dist P, X'\dist P'$.
\end{definition} 
We will consider $X\dist P_{Q}^{model}$, and $X'$ as $X^*_{\rho, \tilde{X}_{X, \sigma}, Q}$ the corresponding denoised image, and will also  call $d(P, P')$ the \textit{expected overlap between $X$ and $X'$}. 
To keep notation simple, for the remainder of this section allow us to write $\tilde{X}$ in place of $\tilde{X}_{X, \sigma}$, with $X$ and $\sigma$ being clear from context.

Our main positive result concerning the choice of $\rho$ is summarized in the following theorem: 

\begin{theorem}\label{thm: rho optimality}
Let $X\dist P_{Q}^{model}$ as in \ref{eq:bm prob} and $\tilde{X}$ be the noisy image. Then choosing $\rho = log\dfrac{1-\sigma}{\sigma}$ to obtain $X^*_{\rho,\tilde{X}, Q }$ is optimal with respect to maximizing the expected overlap between $X$ and $X^*_{\rho,\tilde{X}, Q }$.
\end{theorem}

\begin{proof}
Let $X\dist P_{Q}^{model}$, and $\tilde{X}$ be $X$ afflicted by salt-and-pepper noise of level $\sigma$. 
 Then since $\tilde{X}_{X, \sigma}$ is obtained by flipping pixels with probability $\sigma$, we have the conditional probability 

\begin{equation}
\begin{split}
P_\sigma(\tilde{X} = \tilde{x}| X = x) = & \prod_{i = 1}^v 
\left\{ \sigma (\tilde{x}_i - x_i)^2 +  (1-\sigma) [ 1 - (\tilde{x}_i - x_i)^2 ] \right\}
\\
=& \frac{\exp \left[ -\beta_\sigma \sum_{i=1}^v (\tilde{x}_i - x_i)^2 \right]}{(1+ e^{-\beta_\sigma})^v}, 
\end{split}
\end{equation}
where $\beta_\sigma := \log \frac{1-\sigma}{\sigma}$.
In order to infer the original image $X$ from the noisy one $\tilde{X}$, 
we utilize the Bayes formula 
and calculate the conditional probability $P_{\beta_\sigma, Q}^{\text{post}}(X = x | \tilde{X} = \tilde{x})$.

\begin{equation}\label{eq: P post}
\begin{split}
P_{\beta_\sigma, Q}^{\text{post}}(x|\tilde{x})
=& \frac{P_\sigma (\tilde{X} = \tilde{x} | X= x) P_{Q}^{model} (x)}
{\sum_{\{ x \}} P_\sigma (\tilde{x} | x) P_{Q}^{model} (x)}
\\
=& \frac{\exp \left[ -\beta_\sigma \sum_{i=1}^v (\tilde{x}_i - x_i)^2 -\sum_{i,j=1}^{v+h} Q_{ij} x_i x_j \right]}
{\sum_{\{x\}} \exp \left[ -\beta_\sigma \sum_{i=1}^v (\tilde{x}_i - x_i)^2 -\sum_{i,j=1}^{v+h} Q_{ij} x_i x_j \right]}.
\end{split}
\end{equation}
Note that $x$ includes pixels for hidden nodes, which is fine here. 
Our approach finds the state which is most likely under this distribution, which is realized 
by annealing for the above QUBO with the $\beta_\sigma$ term.\\

The overlap of two vectors $x^*$ and $x$ is given by
\begin{equation}
m(x, x^*)
:= \frac{1}{v + h} \sum_{i=1}^{v + h} (2 x_i - 1) (2 x^*_i - 1),
\end{equation}
the proportion of shared entries.
We consider the average (over the noise) of solutions, $\bar{X}_{\rho, \tilde{x}, Q}$ with
\begin{equation}
(\bar{X}_{\rho, \tilde{x}, Q})_i
= \theta \left( \sum_{ \{x \}}P_{\tilde{Q}}^{model}(x) x_i  - \frac{1}{2} \right),
\label{eq:average_of_solution} 
\end{equation}
where $\theta(x) = 1$ if $x>0$, otherwise 0, noting that the right hand side represents the inferred pixel value based on the expectation from $P_{\tilde{Q}}^{model}$.
We have formally distinguished $P_{\tilde{Q}}^{model}(x)$ from $P_{\rho, Q}^{\text{post}}(x|\tilde{x})$, 
but in fact they are the same.
Note that
\begin{equation}
2 (\bar{X}_{\rho, \tilde{x}, Q})_i - 1
= \text{sign} \left(\sum_{ \{x \}}P_{\tilde{Q}}^{model}(x)(2 x_i - 1) \right), 
\end{equation}
where $\text{sign}(x)$ is the sign of $x$. Let $\alpha_{\sigma, Q} := - \beta_\sigma \sum_{i} (\tilde{x}_i - x_i)^2 - \sum_{i,j} Q_{ij} x_i x_j$ for conciseness.
In order to evaluate the statistical performance of our method 
with coefficient $\rho$ of penalty term, 
we calcuclate the average of overlap as
\begin{equation}
\begin{split}
M_{\beta_\sigma, Q}(\rho)
:=& \sum_{ \{ \tilde{x} \}, \{x \} } 
P_{\sigma}(\tilde{x} | x) P_{Q}^{model} (x) 
m(\bar{X}_{\rho, \tilde{x}, Q}), x)
\\
=& \frac{1}{(1 + e^{\beta_\sigma})^{v } } \frac{1}{ z} \frac{1}{v + h} \sum_{i}
\sum_{\{ \tilde{x} \}, \{x \}} 
e^{\alpha_{\sigma, Q}} 
[2 (\bar{X}_{\rho, \tilde{x}, Q})_i - 1] (2 x_i - 1).
\end{split}
\end{equation}

A sum in the right hand side of the above equation holds
\begin{equation}
\begin{split}
\sum_{\{x \}}
&e^{\alpha_{\sigma, Q}} 
[2 (\mathbb{E}(X^*_{\rho, \tilde{x}, Q})_i - 1] (2 x_i - 1)
\le  
\left| 
\sum_{\{x \}}
e^{\alpha_{\sigma, Q}} 
[2 (\mathbb{E}(X^*_{\rho, \tilde{x}, Q})_i - 1] (2 x_i - 1)
\right|
\\
\le & 
\left| 
\sum_{\{x \}}
e^{\alpha_{\sigma, Q}} 
(2 x_i - 1)
\right|
=  
\sum_{\{x \}}
e^{\alpha_{\sigma, Q}} 
(2 x_i - 1)
\frac{
\sum_{\{ x' \}}
e^{- \beta_\sigma \sum_{i} (\tilde{x}_i - x'_i)^2 - \sum_{i,j} Q_{ij} x'_i x'_j} 
(2 x'_i - 1)
}
{
\left| 
\sum_{\{ x' \}}
e^{- \beta_\sigma \sum_{i} (\tilde{x}_i - x'_i)^2 - \sum_{i,j} Q_{ij} x'_i x'_j} 
(2 x'_i - 1)
\right|
}
\\
= & 
\sum_{\{x \}}
e^{\alpha_{\sigma, Q}} 
(2 x_i - 1) 
\text{sign} \left(\sum_{ \{x' \}}P_{\tilde{Q}}^{model}(x')(2 x'_i - 1) \right)
= \sum_{\{x \}} 
e^{\alpha_{\sigma, Q}} 
[2 (\bar{X}_{\rho, \tilde{x}, Q})_i - 1] (2 x_i - 1).
\end{split}
\end{equation}
Hence, the averaged overlap holds 
\begin{equation}
\begin{split}
M_{\beta_\sigma, Q}(\rho)
\le& \frac{1}{(1 + e^{\beta_\sigma})^{v } } \frac{1}{Z_{1, Q}} \frac{1}{v + h} \sum_{i}
\sum_{\{ \tilde{x} \}, \{x \}} 
e^{- \beta_\sigma \sum_{i} (\tilde{x}_i - x_i)^2 - \sum_{i,j} Q_{ij} x_i x_j} 
[2 (\bar{X}_{\rho, \tilde{x}, Q})_i - 1] (2 x_i - 1)
\\
= & M_{\beta_\sigma, Q}(\beta_\sigma).
\end{split}
\end{equation}
This inequality means that 
the averaged overlap is maximized when $\rho = \beta_\sigma = \log \frac{1-\sigma}{\sigma}$.
\end{proof}

This theorem is based on a known fact in statistical physics of information processing~\cite{H.Nishimori2001} and translates the fact into the setting of our problem.
Notably, the optimal choice of $\rho$ does {\em not} depend on the distribution of the data, but only on the noise level, for which in many real world cases one may have good estimates. The proof of the theorem also reveals the following corollary: 

\begin{corollary}\label{cor: map recovery}
Under the same assumptions of Theorem \ref{thm: rho optimality}, setting $\rho := \log \frac{1-\sigma}{\sigma}$ makes $X^*_{\rho, \tilde{X}, Q}$ the maximum a posteriori estimator for the original noise-free image $X$. 
\end{corollary}
The corollary follows from observing that the energy function in the numerator of the posterior distribution (\ref{eq: P post}) is exactly (\ref{eq: penalty model}) with $\rho := \frac{1-\sigma}{\sigma}$, noting that minimizing (\ref{eq: penalty model}) is equivalent to maximizing (\ref{eq: P post}). However, this framework allows for additional flexibility in choosing the $\rho$ parameter that is absent in standard MAP estimation. In fact, in sections \ref{sec: robust rho} and \ref{sec: qa testing} we go on to demonstrate that in practice, choosing a larger $\rho$ may be beneficial for robustness of the method.

Though Theorem \ref{thm: rho optimality} derives the optimal choice of $\rho$, it does not give any guarantees that the method will yield an improvement in expected overlap, even under its assumptions. 
 Next, we prove a theorem to show that in the case of visible units being independent of one another, our image denoising method produces in expectation {\em strict} denoising improvements with respect to the expected overlap. 
For $c>0$ and a model distribution $P_{Q}^{model}$ as in \ref{eq:bm prob}, let $\mathcal{I}_c$ be the set of indices $i$ such that $|Q_{ii}| > c$. These indices correspond to components of $X$ that are either 0 or 1 with probability at least $\dfrac{1}{1+e^{-c}}$, depending on whether $Q_{ii}$ is positive or negative, respectively.
\begin{theorem}\label{thm: provable denoising}
Suppose that $Q$ is diagonal, $X \dist P_Q$, and that $\tilde{X}$ is $X$ afflicted by salt-and-pepper noise of level $\sigma$. With $\mathcal{I}_c$ as defined above for $c>0$, setting $\rho \geq \log(\frac{1-\sigma}{\sigma})$, and assuming that $\mathcal{I}_\rho \neq \emptyset$, 
the expected overlap of the denoised image and the true image is strictly larger than the expected overlap of the noisy image and the true image, i.e.
\begin{align}
    \mathbb{E} \left[ \sum \mathbb{I}((X^*_{\rho, \tilde{X}, Q})_i = X_i) \right]&> \mathbb{E} \left[ \sum \mathbb{I}(\tilde{X}_i = X_i) \right].
\end{align}
\end{theorem}

\begin{proof}
  Let $\mathcal{I}_c^0 := \{i \in \mathcal{I}_c: Q_{ii} >0\}, \mathcal{I}_c^1 := \{i \in \mathcal{I}_c: Q_{ii}<0\}$. Intuitively, these are the indices which are likely to be zero or one, respectively. Further, letting $x^{\dag i}$ denote the vector obtained by flipping entry $i$ of x, we have that $|f_Q(x) - f_Q(x^{\dag i})| =Q_{ii} > c$ if and only if $i \in \mathcal{I}_c$.  Hence, this reveals that $x^*$ solves (\ref{eq: penalty model}) by setting $x^*_i = 1 \phantom{x} \forall i \in \mathcal{I}_\rho^1, x^*_i = 0 \phantom{x} \forall i \in \mathcal{I}_\rho^0, \text{ and } x^*_i = \tilde{x}_i $ otherwise, since the value of $f_Q$ of (\ref{eq:cost}) is reduced by more than $\rho$, so that the overall penalized objective (\ref{eq: penalty model}) improves despite the $\rho$ penalty accrued by the pixel flips.
   \\
Now, let $X \dist P_{Q}^{model}.$ Let us compute $P((X^*_{\rho, \tilde{X}, Q})_i = X_i)$. The cases where this happens are: $i\in I_\rho^0$ and $X_i = 0$,  $i\in I_\rho^1$ and $X_i = 1$, or $i \notin I_\rho$ and pixel $i$ was not flipped by the noise. \\
We know that if $i\in \mathcal{I}_\rho^b, P(X_i = b) \geq \dfrac{1}{1+e^{-\rho}}$, for $b \in \{0,1\}$, so $P((X^*_{\rho, \tilde{X}, Q})_i = X_i) \geq \dfrac{1}{1+e^{-\rho}}$ for these. For $i \notin \mathcal{I}_\rho, P((X^*_{\rho, \tilde{X}, Q})_i = X_i) = 1- \sigma$, where $\sigma$ is the probability that the pixel was flipped by the noise. 
On the other hand, $P(\tilde{X}_i = X_i) = 1-\sigma \phantom{x} \forall i.$ We characterize 
\begin{align}
    \mathbb{E} \left[ \sum \mathbb{I}((X^*_{\rho, \tilde{X}, Q})_i = X_i) \right]&> \mathbb{E} \left[ \sum \mathbb{I}(\tilde{X}_i = X_i) \right] \label{ineq: restore well}\\
    \sum P((X^*_{\rho, \tilde{X}, Q})_i = X_i) &> \sum P (\tilde{X}_i = X_i) = n\cdot (1-\sigma)
\end{align}
For the left-hand side, assuming $\mathcal{I}_\rho \neq \emptyset$, we have

\begin{align*}
    \sum P((X^*_{\rho, \tilde{X}, Q})_i = X_i) > \sum_{i \in \mathcal{I}_\rho} \dfrac{1}{1+e^{-\rho}} + \sum_{i\notin \mathcal{I}_\rho}(1-\sigma)
    = |\mathcal{I}_\rho|\cdot\dfrac{1}{1+e^{-\rho}} + (n-|\mathcal{I}_\rho|)(1-\sigma)
\end{align*}
so that (\ref{ineq: restore well}) holds when  

\begin{align}
    &|\mathcal{I}_\rho|\cdot\dfrac{1}{1+e^{-\rho}} + (n-|\mathcal{I}_\rho|)(1-\sigma) \geq n(1-\sigma) \\
    &\iff |\mathcal{I}_\rho| \neq 0 \text{ and } \dfrac{1}{1+e^{-\rho}} \geq 1-\sigma \iff \rho \geq log(\dfrac{1-\sigma}{\sigma}) \text{ and } \mathcal{I}_\rho \neq \emptyset,
\end{align}
and the theorem is proven.
\end{proof}
The assumption that matrix $Q$ is diagonal is equivalent to the components of $X$ being independent, which is not realistic with real data. However, since in the RBM model the visible units are independent conditioned on the hidden units, we still consider this independent case to be informative to the denoising method. In fact, if the hidden states were fixed (or known, or recovered correctly), Theorem \ref{thm: provable denoising} would apply. We leave it as a tantalizing open question to generalize this result beyond the independent case.
The assumption of nonemptiness of $\mathcal{I}_\rho$ is a natural one for the denoising task; indeed, when $\mathcal{I}_\rho$ is empty, no entries of Q are large in magnitude, which is equivalent to the entries of $X$ being close to uniformly distributed. In that case, intuitively of course it should not be possible to guarantee that we can denoise an image well if it looks like noise to begin with.

\subsection{Robust Choice of $\rho$}\label{sec: robust rho}
The optimal choice of $\rho$ as derived in Theorem \ref{thm: rho optimality} relies on the assumption that the observed data comes from the learned distribution, or equivalently that the distribution generating our data has been perfectly learned by the RBM. However, in practice we will always only approximately learn the data distribution. Hence, we do not want to rely too heavily on the exact distribution we have learned when we denoise the images. One may hope to have a more robust method by only changing the value of a pixel when there is some confidence in the model that the pixel should be flipped.
We may thus want to penalize flipping pixels slightly more than we should under the idealistic setting of Theorem \ref{thm: rho optimality}, which corresponds to choosing a larger $\rho$ value than $\log \frac{1-\sigma}{\sigma}$, or equivalently using a smaller $\sigma' < \sigma$ value when setting $\rho := \log \frac{1-\sigma'}{\sigma'}$. We opt for the latter as a means of intentionally biasing $\rho$ to make the approach more robust for application. Figures \ref{fig:MNIST sigma comparisons} and \ref{fig:BAS sigma comparisons} in Section \ref{results} show the effect this proposed robustness modification has, demonstrating indeed that choosing a larger $\rho$ via intentionally using a smaller $\sigma$ yields positive results. If the true noise level is $\sigma$, our experiments demonstrate that setting to roughly $\rho := \frac{1-0.75\sigma}{0.75\sigma}$ has a positive effect on performance.

\begin{section}{Empirical Results}\label{results}

This section contains results from implementing the previously described method and comparing it against other denoising approaches. Datasets and code are available on the first author's GitHub page for the purpose of easy reproducibility. 

\subsection{Results with Quantum Annealing}\label{sec: qa testing}

In this subsection, we present empirical results obtained by implementing our model on a quantum annealer, 
D-Wave's Advantage\_system4.1, 
which has 5000 qubits 
and enables embedding of a complete bipartite graph of size $172\times172$. 
Hence, we use $12 \times 12$ pixel images here so that the visible layer is of size 144. 
We test the method on two different datasets with very differently structured data. 
The first dataset is a $12 \times 12$ version of the well-known MNIST dataset \cite{mnist}, created by downsizing the original dataset with nearest-neighbor image downscaling and binarizing pixels. 
The second dataset we use is a $12 \times 12$ pixel Bars-and-Stripes (BAS) dataset, as has been used in closely related work \cite{koshka_reconstruction, Dixit_2021}, where an 8x8 version was used to accomodate a smaller 2000 qubit machine, D-Wave 2000Q used there. Each image consists of binary pixels with either each row or each column sharing the same values, so that each image consists of either ``bars" or ``stripes". \\
\begin{wrapfigure}{l}{0.5\textwidth}
\includegraphics[width = 6cm]{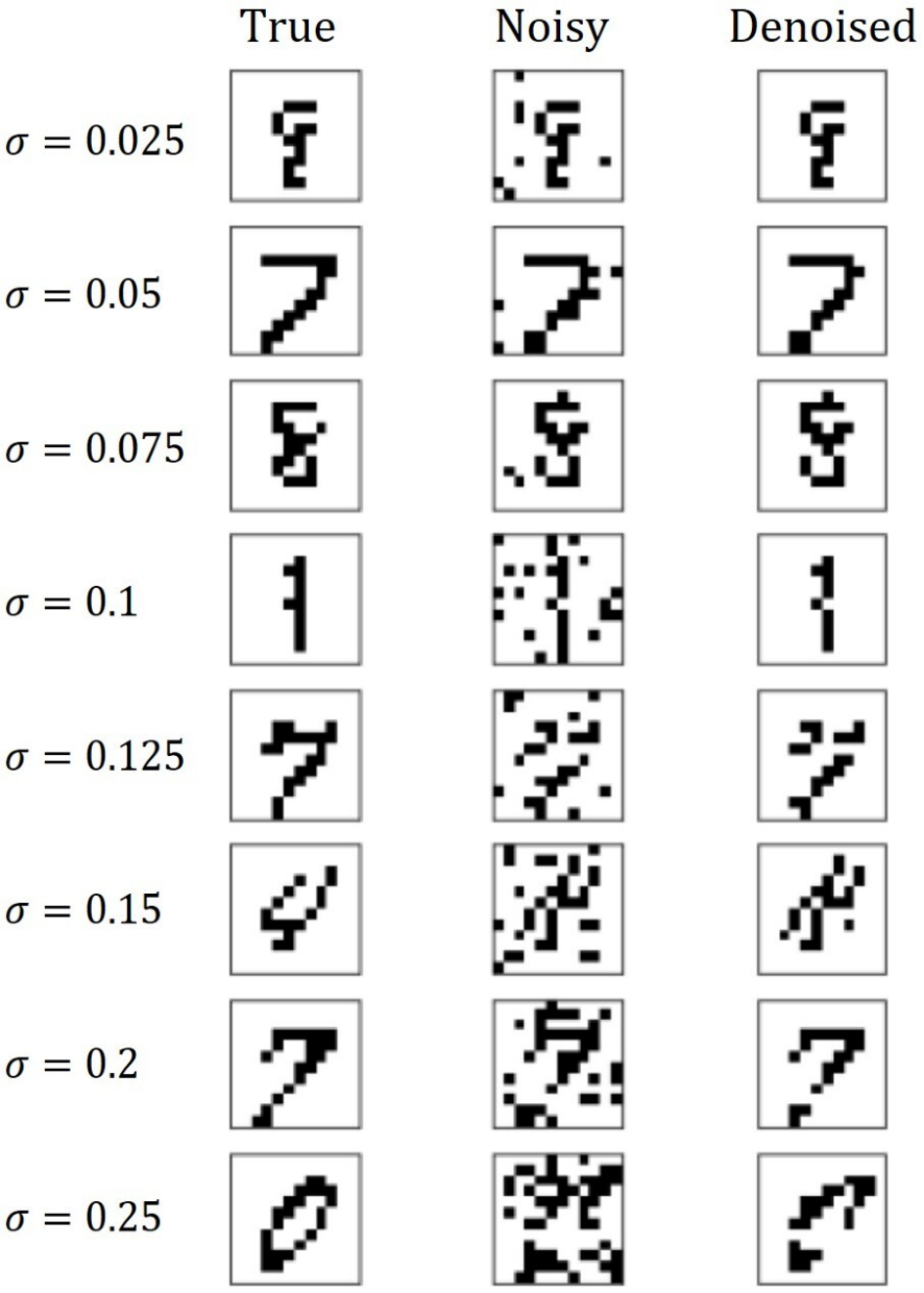}
    \captionof{figure}{Examples of the denoising process using our method showing the true, noisy, and denoised images across different noise levels.}
    \label{fig: denoising visualization}
\end{wrapfigure}
For both datasets we train the RBM by using the classical Contrastive Divergence algorithm first presented in \cite{contrastive_divergence_Hinton2002}. The number of hidden units was set to 50 and 64 for BAS and MNIST, respectively. For the BAS data, 4000 images were generated as training data, and 1000 as test data, while for MNIST, we simply used the full MNIST provided training set of 60,000 images and test set of 10,000 images. 
Noisy images were generated by adding salt-and-pepper noise of level $\sigma$ to images from the test dataset.
Given a noisy image, we are then able to embed and solve the resulting denoising QUBO of \ref{eq: penalty reconstruct} onto a D-Wave quantum annealer, Advantage\_system4.1. 
A function of D-Wave's Ocean softoware, find\_embedding, is utilized to find appropriate mappings from variables in a QUBO to physical qubits on D-Wave's Pegasus graph.
A variable in QUBO is often mapped to multiple physical qubits, called chain, that are strongly connected to each other to behave like a single variable. 
A mapping can be used for every noisy images for each dataset, since their QUBO have the same graph stracuture.
We have prepared in advance 50 sets of the different mappings for each dataset 
and choose a mapping from the pool at random to embed QUBO of each image.
This random selection is done to avoid possible artificial effects on the denoising performance from using only a particular mapping.
Parameters for embedding and annealing, i.e., chain\_strength and annealing\_time, are tuned to maximize the performance.
In particular, we set chain\_strength as the product of a coefficient $c_0$ and the maximum abstract value among the elements of each QUBO matrix, 
where we tune $c_0$.
The adopted values of the parameters are different between MNIST and BAS but the same values for all the range of $\sigma$.
We set ($c_0$, annealing\_time) = (0.6, 50 $\mu$s), (0.5, 40 $\mu$s) for BAS and MNIST, respectively.
The number num\_reads of reads of annealing is 100 for each noisy image.
We calculate the average of solution of each pixel over the reads to approximate Eq.~(\ref{eq:average_of_solution})
and use it to evaluate the overlap that is proportion of pixels in denoised images that matched the original image.
We denoise 200 noisy images for each $\sigma$, which are randomly selected from the pool of test images for each sigma. Note also that for each value of sigma, the different methods compared use the same set of (randomly selected) noisy test images.
Figures \ref{fig:MNIST sigma comparisons} and \ref{fig:BAS sigma comparisons} first investigate the robust choice of $\rho$ as discussed in Section \ref{sec: robust rho}. This is done by using a biased value of $\sigma$ when setting $\rho = \log \frac{1-\sigma}{\sigma}$, instead setting $\rho := \log \frac{1-b\sigma}{b\sigma}$ for some bias factor $b$. The denoising performance for $b\in \{1.25, 1, 0.75, 0.5\}$ are shown, with 95\% confidence intervals obtained by bootstrapping. Note that using a bias factor $b = 1$ means using the true value of $\sigma$ for determining $\rho$.

\setlength{\belowcaptionskip}{10pt}
\begin{minipage}{.45\textwidth}
  \centering
  \includegraphics[width = 7cm]{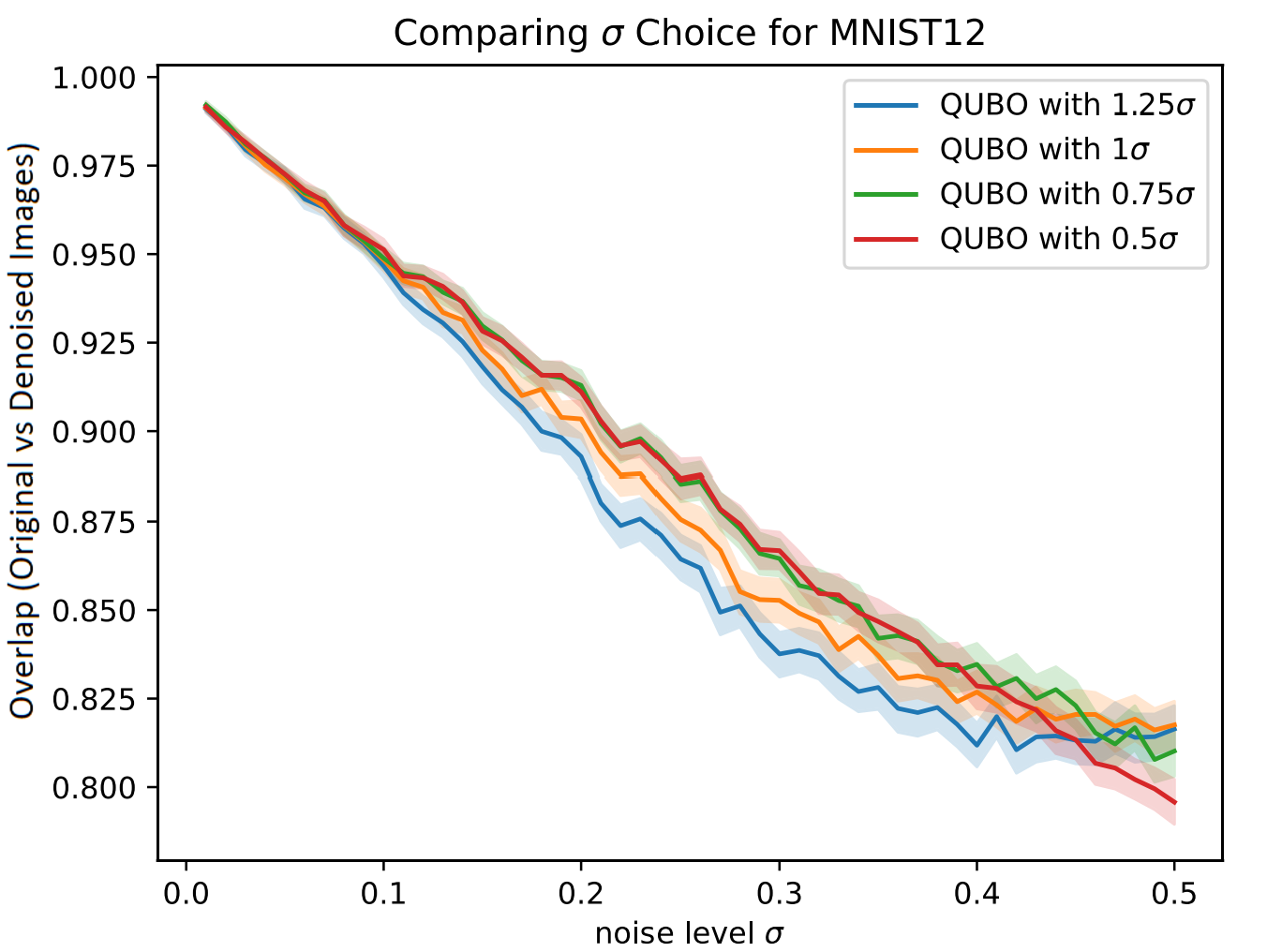}
  \captionof{figure}{Proportion of pixels in denoised MNIST images that matched the original image, for different denoising methods with 95\% CI error bars.  }
  \label{fig:MNIST sigma comparisons}
\end{minipage}
\begin{minipage}{.45\textwidth}
  \centering
  \includegraphics[width = 7cm]{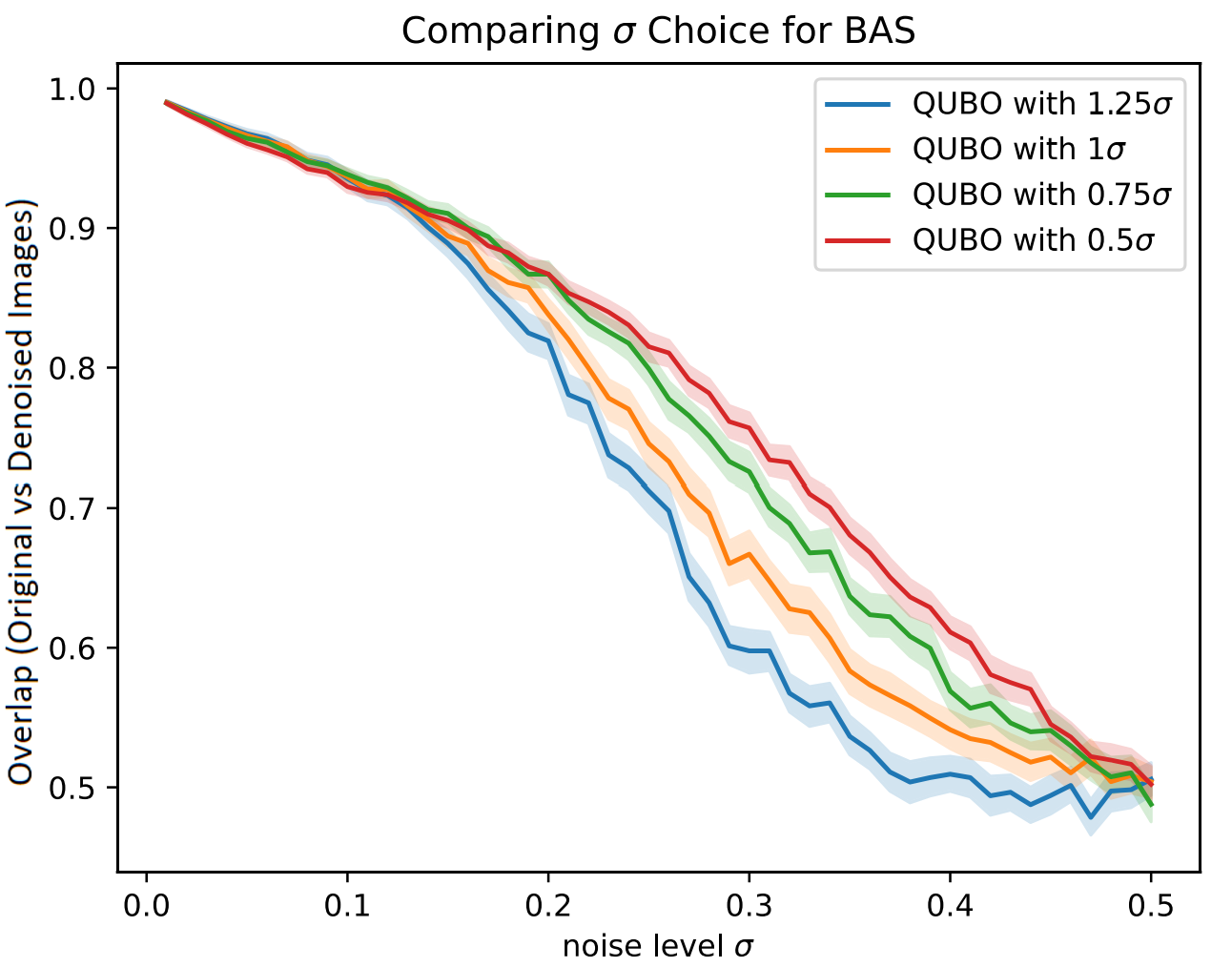}
    \captionof{figure}{Proportion of pixels in denoised BAS images that matched the original image, for different denoising methods with 95\% CI error bars.}
    \label{fig:BAS sigma comparisons}
\end{minipage}%

Based on the empirical performance, using a bias factor of around 0.75 seems to give an improved performance compared to using a bias factor of 1 in both data sets. A bias factor of 0.5 seems to perform quite well across most noise regimes as well, with largely overlapping confidence regions to the 0.75 parameter setting, though in the low-noise setting for the BAS dataset we observe an adverse effect. The authors thus suggest a setting of 0.75 for the bias factor.

Next, in figures \ref{fig:MNIST comparisons} and \ref{fig:BAS comparisons}, we compare our method to popular other denoising methods for binary images on the $12 \times 12$ MNIST and bars-and-stripes datasets, respectively, across different noise levels. When comparing to other methods, a crucial factor is that we choose $\rho$ based off of $\sigma$, but in practice $\sigma$ may be unknown. In light of this, we include two versions of our method in these comparisons. First, we use our method with $\rho := \log \frac{1-\sigma}{\sigma}$, using the true value of $\sigma$ without introducing the recommended bias factor. Secondly, we simulate the situation in which the true $\sigma$ is unknown, and instead we only have a guess for $\sigma$. 
To simulate having an approximate guess for $\sigma$, for each image afflicted by noise of level $\sigma$, we sample $\sigma'$ uniformly from an interval of size $\sigma/2$ centered at sigma. 
We then set $\rho := \log \frac{1-0.75\sigma'}{0.75\sigma'}$, using a bias factor of 0.75 on with this ``guessed" value of $\sigma$. This is a significantly more realistic way of testing our method, since it gives an idea of how well the method may perform when the true noise level present in the noisy images is unknown and must be guessed. Our implementation here only assumes that the practitioner roughly knows the magnitude of the noise. For example, if the true noise is $\sigma = 0.2$, here we sample $\sigma'$ uniformly from $[0.15, 0.25]$ to simulate the guess.

We compare our method to Gibbs denoising with an RBM \cite[section 3.2]{TangHinton2012RobustBM_denoising}, median filtering \cite{Huang1979_median_filter}, Gaussian filtering \cite[chapter 5]{shapiro2001_computervision}, and a graph-cut method \cite{Greig1989_graph_cut_denoising} for denoising. For the Gibbs denoising, we use the same well-trained RBM as for our QUBO-based method, and parameters of the method were carefully tuned for best performance to use 20 Gibbs iterations to then construct the denoised image as the exponentially weighted average of the samples with decay factor 0.8. For the graph-cut method, the recommended parameter setting in the reference of $\beta = 0.5$ is used. 

\setlength{\belowcaptionskip}{10pt}
\begin{minipage}{.45\textwidth}
  \centering
  \includegraphics[width = 7cm]{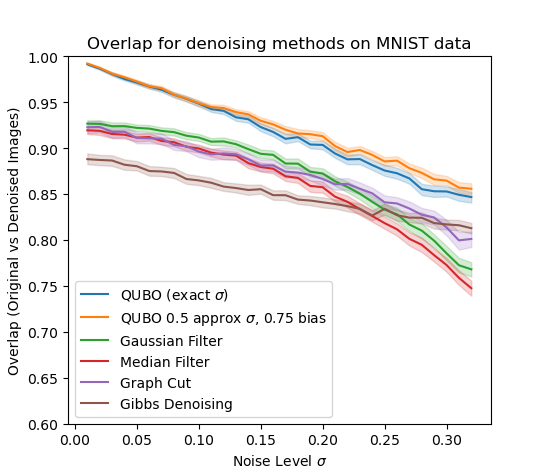}
  \captionof{figure}{Proportion of pixels in denoised MNIST images that matched the original image, for different denoising methods with 95\% CI error bars.  }
  \label{fig:MNIST comparisons}
\end{minipage}
\begin{minipage}{.45\textwidth}
  \centering
  \includegraphics[width = 7cm]{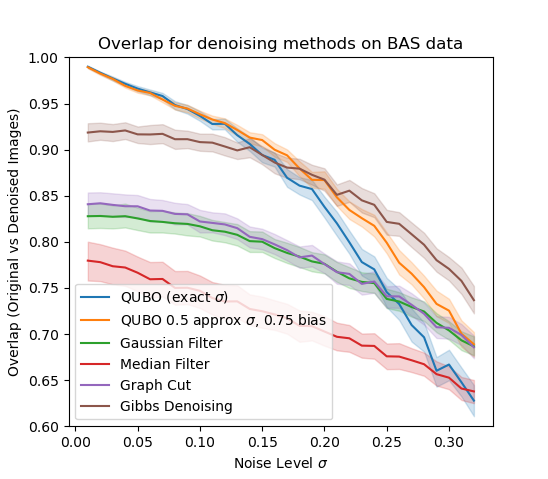}
    \captionof{figure}{Proportion of pixels in denoised BAS images that matched the original image, for different denoising methods with 95\% CI error bars.}
    \label{fig:BAS comparisons}
\end{minipage}%

Overall, the QUBO-based method performs quite strongly. 
Across all noise regimes in the MNIST data, and in most noise regimes in the bars-and-stripes dataset, the method outperforms the others. 
In particular, for the MNIST data the 95\% confidence region for the QUBO method entirely dominates the others. 
Indeed, we see the good performance that our analysis from Section \ref{sec: methods} suggests, even when the true $\sigma$ is unknown and instead guessed. Using a guessed $\sigma$ and the robustness modification of Section \ref{sec: robust rho} makes the method perform as well (if not slightly better) as knowing the true $\sigma$ without the robustness modification.
Only in the noise regime of $\sigma \geq 0.2$ in the BAS data does Gibbs denoising outperform our method.
In Figure \ref{fig: denoising visualization}, we also provide examples of applying our denoising method to noisy images across different noise levels.

\subsection{Testing on Larger Images}\label{sec: mnist testing}

Though we see the the straightforward implementability of our method on quantum annealers as a strong positive, a current drawback on using QAs is the limited data size that can be handled to accomodate their still small qubit capacities. Of course we can still instead test our method on larger datasets by obtaining solutions to the denoising QUBO \ref{eq: penalty model} using other means. In Figure \ref{fig: full size mnist comparisons}, we implement our method on a binarized version of the popular MNIST dataset \cite{mnist} by using simulated annealing \cite{kirkpatrick_SAN_1983} to find solutions to (\ref{eq: penalty model}). We particularly choose to test on the full-size MNIST dataset since we could only use a downscaled version on the QA due to size limitations on the input data, so this experiment serves to test our method without this downscaling. All methods are implemented as described in \ref{sec: qa testing}, and again for our method we use a guessed $\sigma$ to simulate the unknown $\sigma$ case and bias the guess for robustness.

\centering
  \includegraphics[height = 7cm]{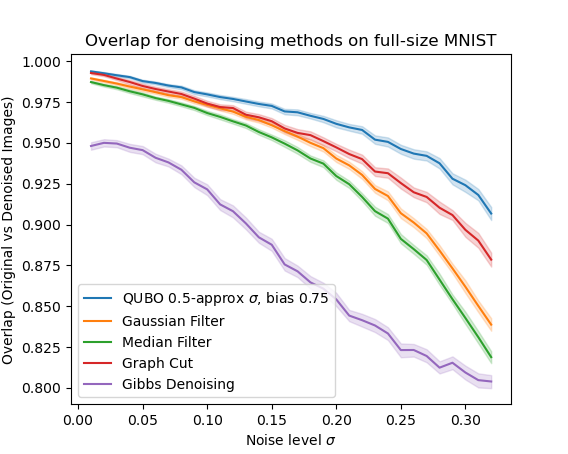}
  \captionof{figure}{Proportion of pixels in denoised images that were correctly denoised, for different denoising methods on the MNIST dataset, with $95\%$ confidence intervals shaded.}
  \label{fig: full size mnist comparisons}

\end{section}
 
\section{Conclusion and Future Work}
We investigated an image denoising framework via a penalty-based QUBO denoising objective that shows promise both theoretically through its statistical properties and practically through its empirical performance together with the proposed robustness modification. The method is well-suited for implementability on a quantum annealer, providing an important application of QAs within machine learning through the fundamental image denoising task. Good results are still obtained on larger datasets when the QUBO is only classically approximated by simulated annealing instead, revealing the approach to be promising even in the absence of QAs. As RBMs form a core building block of many deep generative models such as deep Boltzmann machines or deep belief networks \cite{deep_learning_bengio}, a natural next step is to attempt to incorporate this approach into these more complex models, though current hardware limitations on existing quantum annealers are restrictive. Further, since our method takes advantage of QAs for the denoising step, further research into making use of QAs for the training process of RBMs would yield a full image denoising model where both the model training and image denoising make use of QA.

\section*{Funding}
PK was supported in part by g-RIPS Sendai, Cyberscience Center at Tohoku Univ., and NEC Japan, in early stages of the work. PK is grateful to the USRA Feynman Academy internship program, support from the NASA Academic Mission Services (contract NNA16BD14C), and funding from DARPA under DARPA-NASA agreement SAA2-403688.

\section*{Acknowledgments}
The early stage of this work is based on the work in the g-RIPS Sendai 2021 program. 
The authors thank Y. Araki, E. Escobar, T. Mihara, V. Q. H. Huynh, H. Kodani, A. T. Lin, M. Shirane, Y. Susa, and H. Suito for collaboration in the program.
The authors also acknowledge H. Kobayashi and M. Sato for the use of the computing evironment in the program. P.K. thanks Y. Sukurdeep for helpful feedback and discussions.

{\small
 \setcitestyle{numbers}
 \bibliography{qubocite}
}

\end{document}